\documentclass[11pt,letterpaper,english]{article}
\usepackage[utf8]{inputenc}
\usepackage{authblk}
\title{Self-Adjointness in Klein-Gordon Theory on Globally Hyperbolic Spacetimes}

\author[1]{Albert Much\footnote{much@itp.uni-leipzig.de}}
\author[2]{Robert Oeckl\footnote{robert@matmor.unam.mx}} 
\affil[1]{Institut für Theoretische Physik\\ Universität Leipzig\\ D-04103 Leipzig} 
\affil[1,2]{Centro de Ciencias Matemáticas\\
	Universidad Nacional Autónoma de México\\
	C.P.~58190, Morelia, Michoacán, Mexico}
\date{20 April 2018\\ 12 April 2019 (v2)\\ 17 October 2019 (v3)\\ 26 January 2021 (v4)\\ UNAM-CCM-2018-2}
\usepackage[T1]{fontenc}   
\usepackage{amsmath}  
\usepackage{mathrsfs}
\usepackage{amssymb}          
\usepackage{amsthm}    

\usepackage[colorlinks,pdfpagelabels,pdfstartview = FitH,bookmarksopen = true,bookmarksnumbered =
true,linkcolor = black,plainpages = false,hypertexnames = false,citecolor = black]{hyperref}

\usepackage{geometry}

\newtheorem{theorem}{\textsc{Theorem}}[section]
\newtheorem{lemma}{\textsc{Lemma}}[section]
\newtheorem{proposition}{\textsc{Proposition}}[section]

\newtheorem{corollary}{\textsc{}Corollary}[section] 
\newtheorem{definition}{\textsc{Definition}}[section]

\newtheorem{remark}{Remark}[section]

\newtheorem{assumption}{Assumption}[section]

\newcommand{\R}{\mathbb{R}}

\newcommand{\N}{\mathbb{N}}

\usepackage{fancyhdr}
\numberwithin{equation}{section} 
\begin{document}
\maketitle

\abstract{We prove  essential  self-adjointness of   the spatial part of the linear Klein-Gordon operator with external potential for a large class of  globally hyperbolic manifolds. The proof is conducted by a fusion of new results concerning globally hyperbolic manifolds, the theory of weighted Hilbert spaces and  related functional analytic advances.}

\section{Introduction}

Quantum field theory (QFT) in curved spacetime studies the behavior of quantum fields that propagate in the presence of a classical gravitational field, where the quantum behavior of the gravitational field is neglected. It is seen as an intermediate (and mostly rigorous) step towards a complete theory of quantum gravity (see \cite{K88,WQ,i1,i2,i3} for excellent reviews).

One particularly fruitful context arises from the restriction to globally hyperbolic  spacetimes. The advantage of this class of spacetimes is the existence of a (non-canonical) choice of time, or equivalently the existence of a global Cauchy surface. Then the equations of motion given by the  Klein Gordon equation have a well    posed initial value formulation \cite[Theorem 4.1.2]{WQ}, \cite{K88}.

A well established path towards constructing a free QFT starts with the phase space of initial data of a classical field theory on a spacelike hypersurface. This comes naturally equipped with a symplectic structure. Provided one has  a complex structure, which is compatible (tame) with the symplectic structure, a  ``one particle structure'' is constructed by using the aforementioned objects to induce a positive definite inner product. This Hilbert space is then second quantized into a Fock space \cite{Seg:founddyninf1,Seg:founddyninf2}.
For the case of stationary spacetimes this construction has been done in \cite{A75} and (more rigorously) in \cite{K78}. More recently, the question of finding a complex structure that is compatible with a unitary time-evolution for non-stationary spacetimes has been addressed e.g.\ in \cite{CQ02,AA15} (and references therein). In \cite{T02,CC07} for example such a structure was constructed for the Gowdy cosmology (a time dependent globally hyperbolic spacetime).  

In the present article we   focus on the Klein-Gordon theory. An important ingredient in the construction of   complex structures  or \emph{Hadamard states} (for the case of FRW-spacetime see \cite[Appendix~A]{BF})  in QFT is an essentially self-adjoint    (possibly time-dependent) operator $w^2$. This operator encodes the spatial part of the Klein Gordon equation  (see Equation~\ref{op}). In general this operator appears as a component in the Hamiltonian, that is represented  as a $2\times2$ matrix, see \cite[Equation~2.5]{K78},  \cite[Equation 1.17]{K79} and \cite[Appendix~A]{BF}. Therefore, in order for the Hamiltonian to be self-adjoint (which is of essence for a unitary time-evolution) it is essential for the operator $w^2$ that appears as one  of its components to be self-adjoint.  For a relation of the construction of a complex structure for a time dependent globally hyperbolic spacetime with the  Hamiltonian approach see \cite[Section~VI.B]{CCQ}.
 Moreover, for globally hyperbolic manifolds it was recently realized, see~\cite{DS17}, that essential self-adjointness of the   Klein-Gordon operator is a requirement  for the   construction of various kinds of propagators needed in quantum field theory. 
Propagators are essential for the construction of states and thus for the formulation of QFT in such manifolds by the GNS-construction. A fundamental assumption~\cite[Assumption 1.a.]{DS17} in the proof of essential self-adjointness of the Klein-Gordon operator is precisely the essential self-adjointness of the spatial part of the Klein-Gordon equation.  

In this work we prove that in the case of a globally hyperbolic spacetime, the operator $w^2$ will take the form of a \emph{weighted Laplace-Beltrami operator} (see Equation~(\ref{eq:wlbop}), Proposition~\ref{T3}, Lemma \ref{lem:41} and Theorem~\ref{mt})   plus a potential $V$ (multiplied by a positive smooth function, see Equation~\ref{op}).\footnote{The potential in the Klein-Gordon theory is usually given by $V=m^2+\xi R$, where $m$ denotes the mass of the scalar field and $R$ the scalar curvature of the metric ${g}$.}  The Laplace-Beltrami operator for a metric $\mathbf{h}$   is essentially self-adjoint on $C_{0}^{\infty}(\Sigma)$ if the Riemannian manifold $(\Sigma,\mathbf{h})$, where $\Sigma$ is a Cauchy surface, is geodesically complete\footnote{A complete  Riemannian manifold  is a   Riemannian manifold for which every maximal (inextensible) geodesic is defined on $\mathbb{R}$, see \cite[Definition~1.4.6]{JJ}.},  see \cite{stl,ch1}.  This result was extended in \cite{S01} for the case of weighted Laplace-Beltrami operators and furthermore sufficient  conditions for the potential were given in \cite{S01},  in order for the (weighted) Laplace-Beltrami operator plus a potential to be essentially self-adjoint on $C_{0}^{\infty}(\Sigma)$. In these mentioned works geodesic completeness plays a fundamental role and it will be a necessary condition also for our main result, Theorem~\ref{mt}, to hold. 

In addition to the   essential self-adjointness of the operator $w^2$ we  give a   condition to guarantee  strict positivity thereof. The positivity requirement is imposed in order to take the (unique) square-root and the inverse of the (closure of the) operator which is used in the construction of complex structures.  Proving strict positivity and essential self-adjointness in a rather general setting is the goal of the present paper. 
 
 Kay  proved the essential self-adjointness for the important class of \emph{static} spacetimes, under certain additional boundedness requirements \cite[Theorem 7.2]{K78}.  Our results extend this by showing essential self-adjointness for a considerably larger class of globally hyperbolic spacetimes,  while, in some instances,  dropping at the same time the additional boundedness requirements. To this end we use, as already mentioned,  the   theory of weighted Hilbert spaces (see \cite{AG1}) and related functional analytic advances \cite{S01,BMS}.

Recent results \cite{h2} have also shown an important connection between the existence of an essentially self-adjoint  quantum Hamiltonian (that guarantees unitary evolution) and the requirement for those operators to be bounded from below. Hence, in addition to essential self-adjointness, positivity plays again an important role.

Throughout this work we use Greek letters $\mu, \,\nu=0,\ldots,3$ for spacetime indices and we use Latin letters  $i,\,j,\,k,\ldots$ for  spatial components which run from $1,\ldots,3$.    

\section{Klein-Gordon Theory on Globally Hyperbolic Spacetimes}

We start with basic definitions that are needed for the subsequent results.
\begin{definition}[\textbf{Globally hyperbolic spacetime}]
		We denote a spacetime by $(M, {g})$ where $M$ is a smooth, four-dimensional manifold and $ {g}$ is a Lorentzian metric on $M$ with signature $(-1,+1,+1,+1)$. In addition we assume \emph{time-orientability} of the manifold. This means that there exists a $C^{\infty}$-vector-field $u$ on $M$ that is everywhere timelike, i.e.\ $ {g}(u,u)<0$. A smooth curve
		$\gamma:I\rightarrow M$, $I$ being a connected subset of $\mathbb{R}$, is called \emph{causal} if  $ {g}(\dot{\gamma},\dot{\gamma})\leq0$ where $\dot{\gamma}$ denotes the tangent vector of $\gamma$.  A causal curve is called \emph{future directed} if ${g}(\dot{\gamma},u)<0$ and \emph{past directed} if ${g}(\dot{\gamma},u)>0$ all along $\gamma$ and for a global timelike vector-field $u$. For any point $x\in M$, $J^{\pm}(x)$ denotes the set of all points in $M$ which can be connected to $x$ by a future$(+)/$past $(-)$-directed causal curve. An time-orientable spacetime  is called \emph{globally hyperbolic} if for each pair of points $x,y\in M$ the set $J^{-}(x)\cap J^{+}(y)$ is compact whenever it is non-empty. This definition is equivalent to the existence of a smooth foliation of $M$ in Cauchy surfaces, where a smooth hypersurface of $M$ is called a \emph{Cauchy surface} if it is intersected exactly once by each inextensible causal curve.
	\end{definition}
        
		The advantages of requiring a spacetime to be globally hyperbolic are best displayed by the following    theorem \cite{K79}.
		
		\begin{theorem}
			Given a spacetime $(M, {g})$ the following statements are equivalent: 
			\begin{itemize}
				\item $(M, {g})$ is globally hyperbolic
				\item There exists a (global) Cauchy surface in $(M, {g})$
				\item There exists a choice of time\footnote{See \cite{K79} for definition.} on $(M, {g})$
			\end{itemize}
		\end{theorem}
        
	Hence, an effective way of thinking of a globally hyperbolic spacetimes is to think that these spacetimes admit a choice of time or that they are topologically equivalent to $\mathbb{R}\times\Sigma$ for any Cauchy surface $\Sigma$ (strictly, $\{t\} \times  \Sigma$ for any $t\in\mathbb{R}$) \cite{G70}.
	The authors in \cite{B03} solved a long-standing conjecture by proving that any globally hyperbolic spacetime admits a  \emph{smooth} foliation into Cauchy surfaces \cite[Theorem 1.1]{B03}. Moreover, the induced metric of such a globally hyperbolic spacetime admits a specific form \cite[Theorem 1.1]{B05}.
        
		\begin{theorem}\label{T1}
  Let $(M, 	 {g})$ be a globally hyperbolic spacetime. Then, it is isometric to
	the smooth product manifold $\mathbb{R}\times \Sigma$ with a metric $g$, i.e.,
	\begin{equation} 
       g=-N^2 dt^2+h_{ij}dx^idx^j,
          \label{m}\end{equation}
	where $\Sigma$ is a smooth 3-manifold, $t:\mathbb{R}\times  \Sigma\mapsto\mathbb{R}$   is the natural
	projection, $N: \mathbb{R}\times  \Sigma \mapsto (0,\infty)$ a smooth function, and $\mathbf{h}$ a $2$-covariant symmetric tensor field on $\mathbb{R}\times  \Sigma$, satisfying the following condition:  Each hypersurface  $\Sigma_t\subset M$ at constant $t$ is a Cauchy surface, and the restriction $\mathbf{h}(t)$ of $\mathbf{h}$ to such a  $\Sigma_t$ is a Riemannian metric (i.e.\  $\Sigma_t$ is spacelike).
	\end{theorem} 

                In general $N$ and $\mathbf{h}$ depend on the time and space coordinates. By abuse of notation we also write $N=N(t)$ and $\mathbf{h}=\mathbf{h}(t)$ when denoting the corresponding objects pulled back to $\Sigma$ for a fixed time $t$.

                We proceed to consider the \emph{Klein-Gordon equation} in $(M,  {g})$ with an external \emph{potential} $V$.\footnote{The potential is usually given by $V=m^2+\xi R$, where $m$ denotes the mass of the scalar field and $R$ the scalar curvature w.r.t. the metric ${g}$. The regularity of the potential will be discussed in more detail for the upcoming results.  } A solution $\phi$ satisfies,
\begin{equation}
(\square_{g}- V)\phi=0,
\label{kge}\end{equation} 
where $\square_{{g}}$ is the wave operator with respect to the metric ${g}$. That is, $$\square_{{g}}=  ({\sqrt{|g|}})^{-1}\partial_{\mu}(\sqrt{|g|}g^{\mu\nu}\partial_{\nu}),$$ with $|g|$ denoting the absolute value of the determinant of the metric ${g}$. We denote for each Cauchy surface $\Sigma$ the space of smooth Cauchy data of compact support by  $$\mathscr{S}_{\Sigma}:= C_0^{\infty}(\Sigma)\oplus C_0^{\infty}(\Sigma).$$ Moreover, due to Leray's Theorem~\cite{LJ}, the Cauchy data $\Phi\in\mathscr{S}_{\Sigma}$ given by
\begin{align}
\Phi=\begin{pmatrix}\varphi\\\pi\end{pmatrix}= \begin{pmatrix}\phi\\n^{\mu}\nabla_{\mu}\phi\end{pmatrix}\bigg\rvert_{\Sigma}= \begin{pmatrix}\phi\\N^{-1}\partial_t\phi\end{pmatrix} \bigg\rvert_{\Sigma}
\end{align}
define a unique solution $\phi$ in $C^{\infty}(M)$, where $N^{-1}\partial_{t}=n^{\mu}\nabla_{\mu}$ and the vector $n^{\mu}$ is the future pointing unit vector normal to the surface $\Sigma$.  	Let us further define the   corresponding 		symplectic form  $\Omega_{\Sigma}:\mathscr{S}_{\Sigma} \times \mathscr{S}_{\Sigma} \mapsto \mathbb{R} $ on $\mathscr{S}_{\Sigma}$ as 
\begin{align}\nonumber
 \Omega_{\Sigma} ( \Phi_1, \Phi_2)&=\int_{\Sigma}\left(
 \pi_1 \varphi_2 - \pi_2\varphi_1 
 \right)\sqrt{|\mathbf{h}|}d^3x,
\end{align}
where the integral exists since the functions $\Phi_1$ and  $\Phi_2$ are compactly supported on $\Sigma$. Moreover, for given global solutions of the Klein-Gordon equation (\ref{kge}) the integral is independent of the choice of the Cauchy surface. The symplectic structure $\Omega_{\Sigma}$ makes $\mathscr{S}_{\Sigma}$ into a symplectic vector space (a fact that relies on the requirement that the spacetime is globally hyperbolic, see \cite{Y92}).   Next, let the operator $w^2$ be given by  
	\begin{align}\label{op}
	w^2&=-\frac{N}{\sqrt{|\mathbf{h}|}}\partial_i(
	\sqrt{|\mathbf{h}|} N h^{ij}\partial_j
	) +N^2\,V\\\nonumber&=
	-N^2(\Delta_{\mathbf{h}} -V)-Nh^{ij}\partial_iN\partial_j ,
	\end{align}	
where $\Delta_{\mathbf{h}}$ is the Laplace-Beltrami operator with respect to the associated spatial metric $\mathbf{h}$. The operator  $w^2$ is defined such that the Klein-Gordon equation takes the form, 
	$$(\partial_t^2+f(t,x)\partial_t+w^2)\phi=0,$$
where $f(t,x)=- {N}^{-1}\partial_tN+ (\sqrt{|\mathbf{h}|})^{-1}\partial_t \sqrt{|\mathbf{h}|}$.

\section{Weighted Manifolds and Essential Self-Adjointness}

We proceed by introducing the notion of weighted manifolds and weighted Hilbert spaces. For further details we direct the reader to the excellent reference \cite{AG1}. We begin this section with the  following definition, \cite[Chapter 3.6, Definition 3.17]{AG1}.

\begin{definition}
  	A triple $(\Sigma,	\mathbf{h},\mu)$ is called a \emph{weighted manifold}, if $(\Sigma,	\mathbf{h})$ is a Riemannian manifold and $\mu$ is a measure on $\Sigma$ with a smooth and everywhere positive density function $\rho$, i.e., $d\mu=\rho\,\sqrt{|\mathbf{h}|}\,d^3x$.   	A \emph{weighted Hilbert space}, denoted by $L^2(\Sigma, \mu)$, is given  as the space of all  square-integrable functions  on the manifold $\Sigma$ with respect to the measure $\mu$.  	The corresponding  \emph{weighted Laplace-Beltrami operator} (also called the Dirichlet-Laplace operator),  denoted by $\Delta_{ \mu}$ is,  \begin{equation}\label{eq:wlbop}\Delta_{ \mu}=\frac{1}{\rho \sqrt{|\mathbf{h}|}}\partial_i(\rho\sqrt{|\mathbf{h}|}h^{ij}\partial_j).
  	\end{equation} 
\end{definition}
  
We use the following proposition about weighted manifolds in the subsequent discussion \cite[Chapter 3, Exercise 3.11]{AG1}.

\begin{proposition}\label{T3}
 	Let $a ,\,b $ be   smooth and everywhere positive functions on a weighted manifold $(\Sigma,\mathbf{h},\mu)$ and define a new metric $\tilde{\mathbf{h}}$ and measure $\tilde{\mu}$ by $$\tilde{\mathbf{h}}=a\, {\mathbf{h}},\qquad \mathrm{and} \qquad d\tilde{\mu}=b\, d{\mu}.$$ Then, the weighted Laplace-Beltrami operator $\tilde{\Delta}_{ \tilde{\mu}}$ of the weighted manifold $(\Sigma, \tilde{\mathbf{h}}, \tilde{\mu})$ is given by 	
 $$\tilde{\Delta}_{ \tilde{\mu}} =\frac{1}{b} \,\mathrm{div}_{\mu}(\frac{b}{a}\nabla),$$
where in local coordinates  the divergence of a vector field $v$ is given by 
$$\mathrm{div}_{\mu}v=\frac{1}{\rho}\frac{\partial}{\partial x^{i}}(\rho v^i).$$ In particular, if $a=b$ then
$$\tilde{\Delta}_{ \tilde{\mu}}=\frac{1}{a}{\Delta}_{ {\mu}}.$$ 
\end{proposition}
\begin{proof} 
 	The proof is conducted by using functions of compact support  and the Green formula \cite[Chapter 3, Theorem 3.16]{AG1},
 	\begin{align*}&
 	\int\limits_{\Sigma} 
 	u(\tilde{\Delta}_{ \tilde{\mu}}v)\,d\tilde{\mu}=-
 	\int\limits_{\Sigma} (
 \partial_i \,u)\tilde{h}^{ij} (
 \partial_j \,v) d\tilde{\mu}\\=-&
 \int\limits_{\Sigma} (
 \partial_i \,u)\,\frac{b}{a}\,{h}^{ij} (
 \partial_j \,v) d{\mu}= \int\limits_{\Sigma} u\,\, \mathrm{div}_{\mu}(\frac{b}{a}\nabla  \,v) \,d{\mu}\\= &
\int\limits_{\Sigma} u\,\, \frac{1}{b}\,\mathrm{div}_{\mu}(\frac{b}{a}\nabla  \,v) \,d{\tilde{\mu}} ,	\end{align*}
 for all $u,v\in C_{0}^{\infty}(\Sigma)$. For $b=a$ we have, 
 \begin{align*}&
 \int\limits_{\Sigma} 
 u(\tilde{\Delta}_{ \tilde{\mu}}v)\,d\tilde{\mu}=
 \int\limits_{\Sigma} u\,\, \frac{1}{a}\,\mathrm{div}_{\mu}( \nabla  \,v) \,d{\tilde{\mu}} =
 \int\limits_{\Sigma} u\,\,  (\frac{1}{a}\, {\Delta}_{ {\mu}} \,v) \,d{\tilde{\mu}} .	\end{align*}
\end{proof}

\begin{remark}\label{rem:pos}
From the Green Formula for a weighted Laplace-Beltrami operator it follows that the  weighted Laplace-Beltrami operator  is a positive operator (see \cite[Lemma 4.4 and Equation 4.14]{AG1}).
\end{remark}

Before proceeding to our general result we need to mention another theorem that we use. First, define a local $L^2(\Sigma, {\mu})$ function $f$ as a function that is square integrable (with respect to the scalar product of the weighted Hilbert space $L^2(\Sigma, {\mu})$) on every compact subset of the manifold $\Sigma$ and we write $f\in L^2_{loc}(\Sigma, {\mu})$. Moreover, a symmetric operator $H$ is semi-bounded on $C_{0}^{\infty}(\Sigma)$ if there exists a constant $C\in\mathbb{R}$ such that, 
\begin{equation*}
\langle \Psi,H\,\Psi \rangle\geq -C\langle \Psi, \Psi \rangle,\qquad  \Psi\in C_{0}^{\infty}(\Sigma) .
\end{equation*}
 Then, the  theorem of Shubin  states the following \cite[Theorem 1.1]{S01} (see also \cite{BMS} for an  extension of this result to singular potentials). 
\begin{theorem}\label{T5}
	Let the Riemannian manifold $(\Sigma,\mathbf{h})$ be complete and let the potential  $V\in L^2_{loc}(\Sigma, {\mu})$ be such that we can write $V = V_+ + V_-$, where $V_+\in L^2_{loc}(\Sigma, {\mu})\geq 0$ and $V_-\in L^2_{loc}(\Sigma, {\mu})\leq 0$
	point-wise. Furthermore, let the   operator  
	$H_V=-{\Delta}_{{\mu}}+V$
	be 	semi-bounded from below. Then, the operator    $H_V$ is an essentially self-adjoint operator on $C_0^{\infty}(\Sigma )$.
\end{theorem} 
This theorem is in particular an extension of the result that the (unweighted) Laplace-Beltrami operator for a metric $\mathbf{h}$   is essentially self-adjoint on $C_{0}^{\infty}(\Sigma)$ if the Riemannian manifold $(\Sigma,\mathbf{h})$ is geodesically complete, see \cite{stl,ch1},   to the case of a weighted Laplace-Beltrami operator (\cite[Theorem 6.1]{S01}).

\section{Main Result on Self-Adjointness}

In this section we present our main result, i.e., proving essential self-adjointness and positivity of the operator $w^2$ that was given by formula (\ref{op}). In contrast to previous works, our result also covers the case of globally hyperbolic spacetimes with \emph{unbounded}~$N$ (see next section).

We consider the weighted manifolds $(\Sigma,\mathbf{h},\mu)$ and $(\Sigma, \tilde{\mathbf{h}}, \tilde{\mu})$, where $d\tilde{\mu}= N^{-2}d\mu$ and $\tilde{\mathbf{h}}=N^{-2}\mathbf{h}$. The measure $d\tilde{\mu}=N^{-1}\sqrt{|\mathbf{h}|}\,d^3x$ represents the   one usually considered  in field theory in curved spacetimes to define the  symplectic structure or a real inner product  (see \cite[Equation~4.2.6]{WQ} or \cite{A75}). Hence, one is in general interested in proving the essential self-adjointness of the operator $w^2$ with respect to the measure $d\tilde{\mu}$.  
 
\begin{lemma}\label{lem:41}	Let the weighted manifold  $(\Sigma,\mathbf{h},\mu)$ be given by the  Riemannian manifold  $(\Sigma, \mathbf{h})$  and  the corresponding measure $\mu$ with the smooth strictly positive density function $N$. Moreover, let the weighted manifold $(\Sigma, \tilde{\mathbf{h}}, \tilde{\mu})$ 
	be given by the  Riemannian manifold  $(\Sigma,  \tilde{\mathbf{h}})$, with Riemannian metric $\tilde{\mathbf{h}}=N^{-2}\mathbf{h}$   and    corresponding measure $\tilde{\mu}$ with the smooth strictly positive density function $N^{-1}$.  Then, the operator 	$w^2$ (from Equation~\ref{op}) is given by the sum of a \emph{weighted Laplace-Beltrami operator} and a  potential term,
\begin{align*}
w^2&=-N^2\Delta_{  {\mu}}+N^2\,V\\&=
-\tilde{\Delta}_{  \tilde{\mu}}+N^2\,V,
\end{align*}
 where $\Delta_{  {\mu}}$ and $\tilde{\Delta}_{  \tilde{\mu}}$ are the weighted Laplace-Beltrami operators of the weighted manifolds  $(\Sigma,\mathbf{h},\mu)$ and   $(\Sigma, \tilde{\mathbf{h}}, \tilde{\mu})$, respectively.
\end{lemma} 
\begin{proof} 
	We study the Laplace-Beltrami part of the operator $w^2$, i.e., the operator $N^{-2}w^{2}-V$. We write it as a weighted Laplace-Beltrami operator of the weighted manifold
	$(\Sigma, \mathbf{h}, \mu)$, where the measure is
	$d\mu=N\,\sqrt{|\mathbf{h}|}\,d^3x$. Positivity and smoothness of the function $N$ follow from global hyperbolicity (see Theorem~\ref{T1}) and hence we have
	$$N^{-2}w^{2}-V=-\Delta_{ {\mu}} =-\frac{1}{N \sqrt{|\mathbf{h}|}}\partial_i(N\sqrt{|\mathbf{h}|}h^{ij}\partial_j) .$$  Next, we use Proposition~\ref{T3} and make the following transformations
	\begin{align*}
	\tilde{\mathbf{h}}=N^{-2}\mathbf{h},\qquad \qquad d\tilde{\mu}=N^{-2}d\mu=N^{-1}\,\sqrt{h}\,d^3x.
	\end{align*}
	Note that the operator $N^{-2}$ satisfies the conditions of Proposition~\ref{T3} (the conditions on $a$) as a multiplication operator since global hyperbolicity demands from the operator
	$N$ to be smooth, positive and invertible.
	After applying the transformations we obtain the weighted Laplace-Beltrami operator of the weighted manifold $(\Sigma, \tilde{\mathbf{h}}, \tilde{\mu})$,
	$$\tilde{\Delta}_{  \tilde{\mu}}=N^2\Delta_{  {\mu}},$$
	which is exactly the Laplace-Beltrami term (i.e., the  term  without the potential) of   the operator $w^2$.
\end{proof}
  By using the previous lemma we are able to obtain the following general result.  

\begin{theorem}\label{mt}  
	Let the Riemannian manifold $(\Sigma, \tilde{\mathbf{h}})$ (from Lemma \ref{lem:41}) be   complete and let  the scaled potential     $N^2V = V_+ + V_-$, where $V_+\in L^2_{loc}(\Sigma, \tilde{\mu})\geq 0$ and $V_-\in L^2_{loc}(\Sigma, \tilde{\mu})\leq 0$
	point-wise.  
	Furthermore, let the operator $w^2=-\tilde{\Delta}_{  \tilde{\mu}}+N^2\,V$ be semi-bounded from below. Then, the operator $w^2$
	 is   essentially self-adjoint  on $C_0^{\infty}(\Sigma )\subset L^{2}(\Sigma,\,\tilde{\mu})$.   
\end{theorem}
\begin{proof}
 By Lemma \ref{lem:41} we know that the operator $w^2$ takes the form of a weighted Laplace-Beltrami operator of the weighted manifold $(\Sigma, \tilde{\mathbf{h}}, \tilde{\mu})$ with the addition of a  scaled potential term, i.e.\
 \begin{align*}
 w^2=
 -\tilde{\Delta}_{  \tilde{\mu}}+N^2\,V.
 \end{align*}
 By demanding geodesic completeness of $(\Sigma, \tilde{\mathbf{h}})$,  by imposing the requirements of  local integrability  on the scaled potential term $N^2\,V$ and by the semi-boundedness condition on the operator $w^2$, essential self-adjointness follows from Theorem \ref{T5}.   
\end{proof}
 \begin{proposition}
 	Let the Riemannian manifold $(\Sigma, \tilde{\mathbf{h}})$ (from Lemma \ref{lem:41}) be   complete and let  the   potential $V$ be strictly positive, i.e.\ $V >   \epsilon$ for some $\epsilon>0$. 	Moreover, let  the scaled potential   $N^2V\in L^2_{loc}(\Sigma,\tilde{\mu})$ be locally square integrable. Then, the operator $w^2$ 	is   essentially self-adjoint  on $C_0^{\infty}(\Sigma)\subset L^{2}(\Sigma,\,\tilde{\mu})$ 
 	and the closure of the operator is strictly positive and invertible. 
 \end{proposition}
   \begin{proof}
By Theorem \ref{T1} we know that the function $N$ is positive and hence the condition of strict positivity on the potential   guarantees the semi-boundedness from below by a positive constant ${c>0}$,   i.e.\ 
	\begin{equation}\label{c}
	\langle \Phi, w^2 \Phi\rangle   = 	\langle \Phi, (-\tilde{\Delta}_{  \tilde{\mu}}+N^2V) \Phi\rangle\geq c\Vert \Phi\Vert^2, 
	\end{equation}
	which holds $\forall \Phi \in C_0^{\infty}(\Sigma)$. Hence, by the completeness of the Riemannian manifold $(\Sigma, \tilde{\mathbf{h}})$, by the local integrability condition of the scaled potential and by the semi-boundedness of the operator $w^2$ essential self-adjointness follows from Theorem \ref{T5}. 	Since the operator $w^2$ is a strictly positive  essentially self-adjoint operator, it has only one semi-bounded  self-adjoint extension which coincides with the Friedrich extension (\cite[Theorem X.23, Theorem X.26]{RS2}) that is bounded by the same constant $c$ (see Equation~\ref{c}).
\end{proof}
It follows from this proof that the square root (for $V>\epsilon$ for some $\epsilon>0$) and its inverses are well-defined and given as a unique self-adjoint operator. The self-adjointness of the respective operators is proven    by using   the spectral theorem (\cite[Chapter~V.III.3]{RS1}) or by using \cite[Theorem 3.35]{sr2}, \cite{sr1} or for a new shorter proof see  \cite{sr3}.
\begin{remark} Since the operator $w^2$ is a  positive symmetric operator (with the additional requirement on the potential) we could have used the Friedrich extension (as suggested in \cite{A75}) to prove that there exists a self-adjoint extension. However, the problem with this approach is lack of uniqueness. In particular,  the extension given by the Friedrich theorem is, although self-adjoint, not necessarily unique. Proving essential self-adjointness for a strictly positive operator, on the other hand, implies the uniqueness of the Friedrich extension.
\end{remark}

\section{Complete Riemannian Manifolds}

In the last section we made  the assumption that the Riemannian manifold $(\Sigma,\tilde{\mathbf{h}})$ is complete in order to use Theorem~\ref{T5} to  prove essential self-adjointness.  This section gives concrete results (and examples) for this condition. 

Consider first the case where the hypersurface $\Sigma$ used in the foliation of Theorem~\ref{T1} is compact. From compactness it follows that the Riemann manifold $(\Sigma,\mathbf{h})$ is complete (see \cite[Theorem~1.4.7]{JJ} or \cite[Corollary~5.4]{K78}). But the same is also true for the Riemannian manifold $(\Sigma,\tilde{\mathbf{h}})$ with the conformally transformed metric $\tilde{\mathbf{h}}=N^{-2} {\mathbf{h}}$. Hence, in the case of compact Cauchy surfaces our proof applies without restriction.

We proceed to consider the case where $\Sigma$ is non-compact.

\subsection{Static Globally Hyperbolic Spacetimes}
First, we investigate the  case of static globally hyperbolic spacetimes, i.e.,   spacetimes (of the form given in Theorem \ref{T1}) where $N$ and all components of $\mathbf{h}$ are time independent. The proposition that we use in the subsequent discussion is the following \cite[Proposition~5.3]{K78}.
\begin{proposition}
	The   Riemannian manifold  $(\Sigma,\tilde{\mathbf{h}})$ is Cauchy with respect to the metric tensor 
	\begin{equation}
	g=-dt^2+\tilde{h}_{ij}(\vec{x})dx^idx^j
	\end{equation} if and only if the 	Riemannian manifold  $(\Sigma,\tilde{\mathbf{h}})$ is complete. 
\end{proposition} 
By using the previous proposition we have the following result. 
\begin{corollary}\label{cor:stat}
 Let $(M=\R\times\Sigma,g)$ be a static globally hyperbolic spacetime  given by
 \begin{align*}
 g=-N^2(\vec{x})dt^2+h_{ij}(\vec{x})dx^idx^j.
 \end{align*}
 Then, the   Riemannian manifold $(\Sigma,\tilde{\mathbf{h}}=N^{-2}\mathbf{h})$ is Cauchy (with respect to the metric $\tilde{ {g}}=N^{-2} {g}$) and therefore geodesically complete.
\end{corollary}
\begin{proof}
 Let us consider the  conformally equivalent    spacetime  $(M,\tilde{ {g}})$ that is ultra-static with metric
	\begin{align}\label{usst}
g=-dt^2+N^{-2}h_{ij}dx^idx^j.
	\end{align}
The manifold $(M,\tilde{g})$ is globally hyperbolic since conformal transformations (that are smooth and strictly positive) preserve the causal structure, see \cite{K78,HAW,KIM} and \cite[Appendix~D]{WA}.\footnote{In fact, if the spacetimes $(M, {g})$ and $(M,\tilde{g})$ have identical causal structures, the respective metrics must be  related by   a conformal transformation.} Hence,  by the previous proposition (and the definition of globally hyperbolic manifolds)     the Riemannian manifold  $(\Sigma,\tilde{\mathbf{h}})$ is Cauchy (for the metric $\tilde{{g}}$) and therefore necessarily complete.  
\end{proof}
	 Hence, our proof of essential self-adjointness of $w^2$ (see Theorem~\ref{mt}) holds for \emph{all} static globally hyperbolic spacetimes (of the form given in Theorem~\ref{T1} and a potential that is semi-bounded from below that in addition satisfies the required conditions in \cite{S01}).\footnote{ It can be easily seen that our result contains \cite[Theorem~A.1.]{BF} as a special case.}

       \subsection{Globally Hyperbolic Spacetimes}
	Next, we provide a sufficient condition for a (not necessarily stationary) globally hyperbolic spacetime  $(M=\R\times\Sigma, {g})$, with     Cauchy surface $\Sigma$ (strictly, $\{t\} \times  \Sigma$ for any $t\in\mathbb{R}$) and induced metric $\mathbf{h}$,   to admit a conformally transformed geodesically complete Riemannian manifold $(\Sigma, \tilde{\mathbf{h}})$. In this context we use a fundamental theorem, \cite[Theorem~2.1]{ghc2} (see also \cite{ghc1,ghc3}),  for general spacetime manifolds $(M,g)$.\footnote{The authors give the theorem for a metric with shift vector. However, if the shift vector is equal to zero the condition of the Theorem concerning the shift vector is trivially satisfied.}
	  Let the spacetime $M$ be given by the product $M=\mathbb{R}\times\Sigma$, where $\Sigma$ is an $n$-dimensional smooth manifold. Equip the manifold  $M$  with a $n+1$-dimensional Lorentz  metric $g$ of the form	\begin{align}\label{eq:metric}
g=-N^2(\vec{x},t)dt^2+h_{ij}(\vec{x},t)\, dx^i \, dx^j ,
\end{align}
where $N(\vec{x},t)$ is the lapse function  and $\Sigma_t=\{t\}\times\Sigma$, spatial slices of $M$, are spacelike sub-manifolds equipped with the time-dependent metric $\mathbf{h}_t=h_{ij}(\vec{x},t)dx^i\,dx^j$.\footnote{Such a product space $M$ is called a \textit{sliced space}. The spacetime is time-oriented by increasing $t$.} Moreover, let the following Assumption be met. 
\begin{assumption} \label{ass:bound} With respect to  the components of the metric $g$ we have the following bounds.\newline
	\begin{enumerate}
		\item The lapse function 	 is bounded from above and below for all $t$ by  $\alpha_{B}, \alpha_{C}\in\mathbb{R}$, i.e.,
 $$0<\alpha_{B}\leq N(\vec{x},t)\leq\alpha_{C}.$$
		\item The metric $\mathbf{h}(\vec{x},t)$ is uniformly bounded by the metric $\mathbf{h}(\vec{x},0)$ for all $t\in\mathbb{R}$ and tangent vectors $u\in T\Sigma$. That is, there exists universal constants $A,D\in\mathbb{R}>0$ such that,
		\begin{align}\label{ineq:metAD}
		A\,h_{ij}(\vec{x},	0)u^i\,u^j\leq h_{ij}(\vec{x},t)u^i\,u^j\leq D\,h_{ij}(\vec{x},	0)u^i\,u^j .
		\end{align}
	\end{enumerate}
\end{assumption}

\begin{theorem}[{\cite[Theorem~2.1]{ghc2}}]
  \label{thm:ghst}
	For the spacetime manifold $(M,g)$, with metric $g$ that satisfies Assumption \ref{ass:bound},  the two following statements are equivalent. $\,$\newline
	\begin{enumerate}
		\item $(\Sigma, \mathbf{h})$ is a complete Riemannian manifold. $\,$\newline
		\item The spacetime $(M,g)$ is globally hyperbolic. 
	\end{enumerate}
\end{theorem} 
	Since   with regards to our proof of essential self-adjointness we are interested in the conformally transformed metric $\tilde{\mathbf{h}}$ we obtain the following result. 
\begin{theorem}\label{thm:ghcomp}
  Let the globally hyperbolic spacetime $M$ be given by the product $M=\mathbb{R}\times\Sigma$,    with a $4$-dimensional Lorentz  metric $g$ of the Form~\eqref{eq:metric}. Moreover, let Inequality~(\ref{ineq:metAD}) be fulfilled for the conformally transformed  metric $\tilde{\mathbf{h}}:=N^{-2}{\mathbf{h}}$. Then, the Riemannian manifold $(\Sigma,\tilde{\mathbf{h}})$ is geodesically complete. 
\end{theorem}
	\begin{proof}
	Since conformal transformations do not change the causal structure, the  spacetime $(M,\tilde{g})$ with metric,  
\begin{align*}
\tilde{g} =-dt^2+\tilde{h}_{ij}(\vec{x},t)\, dx^i \, dx^j ,
\end{align*}
		is globally hyperbolic. Therefore, the first bound of Assumption~\ref{ass:bound} is trivially satisfied and the second bound is the demanded condition. Hence,  the proof is concluded by the globally hyperbolicity of $(M,\tilde{g})$ and Theorem~\ref{thm:ghst}.
	\end{proof}

\section{Discussion}

Essential self-adjointness of the spatial part $w^2$ of the Klein Gordon operator has been proven for the case that the conformally transformed  Cauchy surface under consideration is complete, i.e. that $(\Sigma,N^{-2}\mathbf{h})$ is a complete metric space.  In addition we gave a necessary condition on the conformal factor for the manifold $(\Sigma,N^{-2}\mathbf{h})$ to be complete without necessarily having a complete Riemannian manifold $(\Sigma,\mathbf{h})$.

This generalizes a corresponding result \cite[Theorem~7.2]{K78} of Kay, who had to assume   that  $(\Sigma, \mathbf{h})$ is complete and  that $N$ is bounded as a multiplication operator. In particular, if $N$ is bounded by positive constants from above and below, it follows (by strong equivalence) that $(\Sigma,N^{-2}\mathbf{h})$ is a complete metric space. However, the restriction on $N$ and more specifically the assumption that $(\Sigma, \mathbf{h})$ is complete\footnote{Kay was aware of this observation  and gave explicit examples in \cite{K78} for which a conformal transformation   completes an incomplete metric space.} is not necessary as long as the   conformally transformed metric becomes complete. Hence, given a globally hyperbolic manifold that is smoothly foliated by non-complete Cauchy surfaces, the proof of essential self-adjointness of the operator $w^2$ is still possible (by standard methods) if the conformal factor $N^{-2}$ acts as a completion.    In the case of stationary globally hyperbolic spacetimes  of the form given in Theorem~\ref{T1} and a potential that is semi-bounded from below that in addition satisfies the required conditions in \cite{S01}, our main result, Theorem~\ref{mt}, holds without restriction to generality.

Although we require global hyperbolicity of spacetime, a possible application of our theorem to QFT on certain spacetimes that are not globally hyperbolic might be possible. The reason therefore lies in the manner how QFT is studied on such spacetimes.  In particular,  in \cite{K92}  a generalization was proposed in the context of the \emph{algebraic approach} to quantum field theory to study the case of non-globally hyperbolic spacetimes. Basically one  requires that every point in a given spacetime $M$ should have a globally hyperbolic neighborhood $N$. Moreover, the algebra of observables of the spacetime $M$ restricted to $N$ should be equal to the algebra obtained by regarding the neighborhood $N$ as a globally hyperbolic spacetime in its own right (with some choice of time).
With regards to our result this means that we can define a Klein-Gordon equation on each such globally hyperbolic neighborhood and use Theorem~\ref{mt} to prove the essential self-adjointness of the spatial part of the Klein-Gordon operator in that neighborhood.
 
Another approach for doing QFT on not necessarily globally hyperbolic spacetimes is \emph{general boundary quantum field theory} \cite{Oe:gbqft,Oe:feynobs}. There, Hilbert spaces of states are associated to hypersurfaces that need not even be spacelike. However, in many contexts the restriction to spacelike hypersurfaces is sufficient and our result is useful there.

 \section*{Acknowledgments}
 
 The authors would like to thank  Suzanne Lanéry for many useful discussions and one of the authors (AM) would like to thank  Bernard Kay for   enlightening and interesting discussions on quantum field theory in curved spacetimes and related functional analytical questions.   The authors  acknowledge partial support from CONACYT project 259258.
 
\newcommand{\eprint}[1]{\href{https://arxiv.org/abs/#1}{#1}}
\bibliographystyle{stdnodoi}
\bibliography{allliterature1,stdrefsb}

\end{document}